\def\year{2020}\relax
\documentclass[letterpaper]{article} % DO NOT CHANGE THIS
\usepackage{aaai20}  % DO NOT CHANGE THIS
\usepackage{times}  % DO NOT CHANGE THIS
\usepackage{helvet} % DO NOT CHANGE THIS
\usepackage{courier}  % DO NOT CHANGE THIS
\usepackage[hyphens]{url}  % DO NOT CHANGE THIS
\usepackage{graphicx} % DO NOT CHANGE THIS
\usepackage{graphicx}  % DO NOT CHANGE THIS
\usepackage{amsfonts,amsthm,amssymb}
\usepackage{xspace}
\usepackage{mathtools}
\usepackage[obeyFinal]{todonotes}
\usepackage{scalerel}
\usepackage{thmtools}
\usepackage{thm-restate}
\usepackage[ruled,noend]{algorithm2e}
\usepackage{paralist}

\newcommand{\LTL}{\ensuremath{\text{LTL}}\xspace}
\newcommand{\LTLf}{\ensuremath{\text{LTL}_f}\xspace}
\newcommand{\PLTL}{\ensuremath{\text{PLTL}_f}\xspace}
\newcommand{\PLTLz}{\ensuremath{\text{PLTL}_f^0}\xspace}

\newcommand{\At}{\ensuremath{\mathsf{At}}\xspace}

\newcommand{\nats}{\ensuremath{\mathbb{N}}\xspace}
\newcommand{\acc}{\ensuremath{\mathsf{acc}}\xspace}
\newcommand{\branch}{\ensuremath{\mathsf{b}}\xspace}

\newcommand{\ini}{\ensuremath{\mathsf{in}}\xspace}
\newcommand{\ms}{\ensuremath{\mathsf{m}}\xspace}
\newcommand{\wt}{\ensuremath{\mathsf{wt}}\xspace}
\newcommand{\run}{\ensuremath{\mathsf{run}}\xspace}

\newcommand{\swt}{\ensuremath{\mathsf{w}}\xspace}
\newcommand{\csub}{\ensuremath{\mathsf{csub}}\xspace}
\newcommand{\END}{\ensuremath{\mathsf{END}}\xspace}

\newcommand{\atom}{\ensuremath{\mathbf{a}}\xspace}

\newcommand{\Amc}{\ensuremath{\mathcal{A}}\xspace}
\newcommand{\Bmc}{\ensuremath{\mathcal{B}}\xspace}
\newcommand{\Imc}{\ensuremath{\mathcal{I}}\xspace}
\newcommand{\Lmc}{\ensuremath{\mathcal{L}}\xspace}
\newcommand{\Omc}{\ensuremath{\mathcal{O}}\xspace}
\newcommand{\Pmc}{\ensuremath{\mathcal{P}}\xspace}
\newcommand{\Qmc}{\ensuremath{\mathcal{Q}}\xspace}
\newcommand{\Smc}{\ensuremath{\mathcal{S}}\xspace}

\newcommand{\Imf}{\ensuremath{\mathfrak{I}}\xspace}

\newcommand{\PS}{\text{\sc{PSpace}}\xspace}
\newcommand{\ET}{\text{\sc{ExpTime}}\xspace}

\newcommand{\prob}[1]{\ensuremath{\scaleobj{1.2}{\circledcirc}_{#1}}\xspace}

\newcommand{\argmax}{\text{argmax}\xspace}
\newcommand{\declare}{Declare\xspace}

\newcommand{\constraint}[1]{\texttt{#1}}

\newtheorem{theorem}{Theorem}

\newtheorem{proposition}[theorem]{Proposition}
\newtheorem{corollary}[theorem]{Corollary}
\theoremstyle{definition}
\newtheorem{example}[theorem]{Example}
\newtheorem{definition}[theorem]{Definition}

\allowdisplaybreaks

\title{Temporal Logics Over Finite Traces with Uncertainty (Full Version)}
\author{
Fabrizio M.\ Maggi \\ University of Tartu \\ f.m.maggi@ut.ee 
\And
Marco Montali \\ Free University of Bozen-Bolzano \\ montali@inf.unibz.it
\And
Rafael Pe\~naloza \\ University of Milano-Bicocca \\ rafael.penaloza@unimib.it
}
 
\begin{document}
%\nocopyright
\maketitle

\begin{abstract}
 Temporal logics over finite traces have recently seen wide application in a number of areas, from 
 business process modelling, monitoring, and mining to planning and decision making. However, real-life dynamic systems contain a 
 degree of uncertainty which cannot be handled with classical logics. We thus propose a new probabilistic temporal logic over finite 
 traces using superposition semantics, where all possible evolutions are possible,
 until observed. We study the properties of the logic and provide automata-based mechanisms for deriving
 probabilistic inferences from its formulas. We then study a fragment of the logic with better computational properties. Notably, formulas 
 in this fragment can be discovered from event log data using off-the-shelf existing declarative process discovery techniques.
\end{abstract}

\section{Introduction}

Linear temporal logic (\LTL) is one of the most important formalisms to declaratively specify and reason about the evolution of systems and processes~\cite{BaKa08}. Traditionally, \LTL adopts a linear, infinite model of time where formulas are interpreted over infinite traces. In recent years, increasing attention has been given to a different version of the logic, \emph{\LTL over finite traces} or \LTLf \cite{GiVa13}, which adopts instead a finite-trace semantics. From the modelling point of view, \LTLf matches settings where each execution of the system is eventually expected to end (even though there is no bound on the number of steps required to reach the termination point). From the reasoning point of view, the automata-theoretic characterisation of \LTLf relies on classical finite-state automata \cite{GiVa13,GiMM14}, which are easier to manipulate and pave the way for the development of robust and efficient reasoning techniques. In fact, \LTLf and extensions thereof have been widely employed in a number of application domains relevant for AI: from declarative business process modelling \cite{PeSV07,Mon10}, monitoring \cite{MMWV11,DDGM14}, and mining \cite{MaCV12,DMMP17}, to planning \cite{GHL09,GiRu18} and decision making \cite{BrDP18}.

When considering real-world processes, uncertainty (which is inexpressible in classical logics) is unavoidable. For example, some pieces may be defective, external events may delay a service, and the loan may lead to an outcome depending on various implicit factors. Handling these scenarios requires a logical formalism capable of expressing uncertainty. Surprisingly, to the best of our knowledge no probabilistic extension of \LTLf has been considered so far. Although several probabilistic variants of infinite-time temporal logics exist \cite{Ognj06,Mora11,Konu10,KovP18,Paleo16},
the complex interaction of probabilities and time usually requires syntactic or semantic restrictions in the logic, and does not directly carry over the finite-trace setting. 

To overcome both challenges at once, we propose a new probabilistic extension of \LTLf, called \PLTL, that essentially predicates over the possible evolutions of a trace. The main novelty of \PLTL lies in its \emph{superposition} semantics, where
every evolution is possible (with different probabilities) until it is observed. This semantics accommodates a seamless interaction of probabilities and time that was not possible in previous formalisms, and elegantly fits over finite traces. \PLTL is a direct
generalisation of \LTLf: \PLTL formulas without probabilistic constructors are in fact \LTLf formulas.
\PLTL adequately describes probabilistic temporal or dynamic properties of process
executions; e.g., it can express that a shipped package will eventually reach its
destination with probability 0.95, or that a machine will fail in the next 100 timepoints with probability below
0.001. 

Our second main contribution is an investigation of the logical and computational properties of \PLTL, introducing automata-based algorithms for deciding satisfiability of \PLTL formulas and for computing the
most likely executions of a system described in this logic. These core reasoning services provide the basis for sophisticated, domain-specific tasks such as a probabilistic version of conformance checking \cite{CDSW18} and (prefix) monitoring \cite{MMWV11}. Unsurprisingly, due to the intertwined connection of temporal and probabilistic constructors, handling
\PLTL formulas becomes \ET-hard. This leads us to our third contribution: a study of a fragment of \PLTL, called \PLTLz, where the complexity of reasoning falls to \PS in the length of the formula, matching the classical \LTLf case. Notably, formulas in this fragment can be discovered from event log  data using off-the-shelf existing declarative process discovery techniques \cite{MaCV12}.

This manuscript extends the published work~\cite{MaMP-AAAI20} with full proofs and an extended example.

\section{Preliminaries}

We briefly introduce tree and weighted string automata, assuming basic formal language knowledge. 
For more details, see \cite{TATA2007,DrKV-09}.

\smallskip
\noindent 
\textbf{Tree Automata.}
A \emph{tree} is a set of words of natural numbers $T\subseteq \nats^*$, which is closed under prefixes
and preceding siblings; i.e., if $wi\in T$, then $w\in T$, and
$wj\in T$ for all $1\le j\le i$. A tree is \emph{finite} if its cardinality is finite. Each finite tree $T$ has
a maximum number $k\in\nats$ (its \emph{width}) s.t.\ $wk\in T$ for some $w\in\nats^*$. 
The empty word $\varepsilon$ is the \emph{root}, and a \emph{leaf} is a node $w\in T$ s.t.\
$w1\notin T$. A \emph{labelling} of $T$ on a set $\Sigma$ is a mapping $T\to \Sigma$. A tree
with a labelling is a \emph{labelled tree}. A \emph{branch} of the tree $T$ is a sequence $w_1,\ldots,w_n$ of nodes
of $T$ such that $w_1=\varepsilon$, $w_n$ is a leaf node, and for every $i,1\le i< n$, $w_{i+1}=w_im$ for $m\in\nats$.
If $T$ is labelled, we call branch also the sequence of labels of a branch.

A $k$-ary \emph{tree automaton} is a tuple $\Amc=(\Qmc,\Delta,I,F)$ where \Qmc is a finite set of 
\emph{states}, $I,F\subseteq \Qmc$ are the \emph{initial} and \emph{final} states, 
respectively, and $\Delta\subseteq \Qmc\times\bigcup_{i\le k}\Qmc^k$ is the \emph{transition relation}.
A \emph{run} of \Amc on the tree $T$ is a labelling $\rho:T\to \Qmc$ s.t.\ $\rho(\varepsilon)\in I$ and
for every $w\in T$, if $wn\in T$ but $w(n{+}1)\notin T$, then $(\rho(w),\rho(w1),\ldots,\rho(wn))\in\Delta$.
It is \emph{successful} if for every leaf node $w\in T$, $\rho(w)\in F$. The \emph{language
accepted} by \Amc is the set $\Lmc(\Amc)$ of finite trees for which there is a successful run of \Amc. The
\emph{emptiness problem} asks whether $\Lmc(\Amc)=\emptyset$.

The emptiness problem of $k$-ary tree automata is decidable in time 
$\Omc(|\Qmc|^{k+2})$~\cite{VaWo86} by computing \emph{good states};
i.e., those that appear in a successful run. 
States $q\in F$ without transitions are good.
The set of good states is iteratively extended, adding any state that has a transition leading to
only good states. This iteration reaches a fixpoint after checking the transitions of each
state at most $|\Qmc|$ times. As there are at most $|\Qmc|^{k+1}$ transitions, the set of good states 
is computable in time $\Omc(|\Qmc|^{k+2})$.
\Amc is not empty iff at least one initial state is good.
The \emph{reduced automaton} $\dddot\Amc$ of \Amc is obtained by
removing all bad (i.e., non-good) states from \Amc. \Amc and $\dddot\Amc$ accept the same language, and
$\Lmc(\dddot\Amc)\not=\emptyset$ iff $\dddot\Amc$ contains at least one initial state. 

An alternative emptiness test constructs a successful run top-down. It
guesses an initial state to label the root node, and iteratively guesses transitions for every node not labelled with a 
final state. Through a depth-first construction, the algorithm preserves in memory
only one branch at a time, together with the information of the chosen transitions. Since every branch can be 
restricted to depth $|\Qmc|$, the process requires $\Omc(|\Qmc|\cdot k)$ space \cite{BaHP08}. 

\smallskip
\noindent 
\textbf{Weighted Automata.}
Consider the \emph{probabilistic semiring} $A=([0,1],\max,\times)$ with
the usual $\max$ and product on $[0,1]$.
A \emph{weighted automaton} is a tuple $\Amc=(\Qmc,\ini,\wt,F)$ where
\Qmc is a finite set of \emph{states}, $F\subseteq\Qmc$ are the \emph{final states}, $\ini:\Qmc\to[0,1]$
is the \emph{initialization function} and $\wt:\Qmc\times\Qmc\to [0,1]$ is the \emph{weight function}. A 
\emph{run} of \Amc is a finite sequence $\rho=q_0,q_1,\ldots,q_n$ with $q_n\in F$; $\run(\Amc)$ is
the set of all runs of \Amc. The weight of $\rho=q_0.\ldots,q_n\in\run(\Amc)$ is
$\wt(\rho):= \prod_{i=0}^{n-1}\wt(q_i,q_{i+1})$.
The \emph{behaviour} of \Amc is $\|\Amc\|:=\max_{\rho\in\run(\Amc)}\ini(q_0)\cdot\wt(\rho)$.
To compute the behaviour of the weighted automaton \Amc, we adapt the emptiness test for unweighted
automata to consider the weights of the transitions computing a function 
$\swt:\Qmc\to[0,1]$, where $\swt(q)$ is the maximum weight of all runs starting in $q$.
Initialize $\swt_0(q)=1$ if $q\in F$ and $\swt_0(q)=0$ if $q\notin F$.
Iteratively compute $\swt_{i+1}(q):=\max_{q'\in\Qmc}\wt(q,q')\swt_i(q)$. After
polynomially many iterations, we reach a fixpoint where $\swt_i\equiv\swt_{i+1}$.
If $\swt:=\swt_i$, then $\|\Amc\|=\max_{q\in\Qmc}\ini(q)\swt(q)$.

\section{The Probabilistic Temporal Logic \PLTL}
\label{sec:pltl}

\PLTL extends the linear temporal logic on finite traces \LTLf
\cite{GiVa13}, with
a probabilistic constructor expressing uncertainty about the evolution of traces. 
The only syntactic difference between \LTLf and \PLTL is this new 
constructor. Formally, \PLTL formulas are built by the following syntactic rule
where $a$ is a propositional variable, $p\in[0,1]$, and ${\bowtie}\in\{{\le},{\ge},<,>\}$:
\begin{align*}
\varphi ::= {} & a \mid \neg \varphi \mid \varphi\land\varphi \mid \bigcirc \varphi \mid \varphi U \varphi 
			\mid \prob{\bowtie p}\varphi.
\end{align*}
Intuitively, 
$\prob{\bowtie p} \varphi$ means that, at the next point in time, $\varphi$ holds with probability $\bowtie p$.
To formalise this, we use tree-shaped interpretations providing a class of alternatives branching 
into the future in a new \emph{superposition semantics}.
A \emph{\PLTL interpretation} is a triple $I=(T,\cdot^I,P)$, where $T$ is a finite tree, $\cdot^I$ is a 
labelling of $T$ on the set of propositional valuations,%
\footnote{As usual, we describe a propositional valuation by the set of variables it makes true.}
and $P:T\setminus\{\varepsilon\}\to[0,1]$ is s.t.\ for all $w\in T$, $\sum_{wi\in T}P(wi)=1$.
Satisfiability of a formula in a tree node is defined inductively, extending the \LTLf semantics.
For an interpretation $I=(T,\cdot^I,P)$ and $w\in T$:
\begin{compactitem}[$\bullet$]
\item $I,w\models a$ iff $a\in w^I$
\item $I,w\models\neg\varphi$ iff $I,w\not\models\varphi$
\item $I,w\models\varphi\land\psi$ iff $I,w\models\varphi$ and $I,w\models\psi$
\item $I,w\models\bigcirc\varphi$ iff $w$ is not a leaf node and for all $i\in\nats$, if $wi\in T$ then 
	$I,wi\models\varphi$
\item $I,w\models\varphi U\psi$ iff either (i) $I,w\models\psi$ or (ii) $I,w\models\varphi$ and for all $i\in\nats$, if
	$wi\in T$ then $I,wi\models \varphi U\psi$
\item $I,w\models\prob{\bowtie p}\varphi$ iff $\sum_{wi\in T; I,wi\models\varphi}P(wi)\bowtie p$.
\end{compactitem}
$I$ is a \emph{model} of $\phi$ if $I,\varepsilon\models\phi$. $\phi$ is \emph{satisfiable} if it has a model.

\begin{example}
\label{exa:form}
Figure~\ref{fig:exa0} shows three models of the formula $\phi_0:=\prob{\le 0.5}a\land\prob{\ge 0.6}\bigcirc b$. 
\begin{figure}
\centering
\includegraphics[width=0.85\columnwidth]{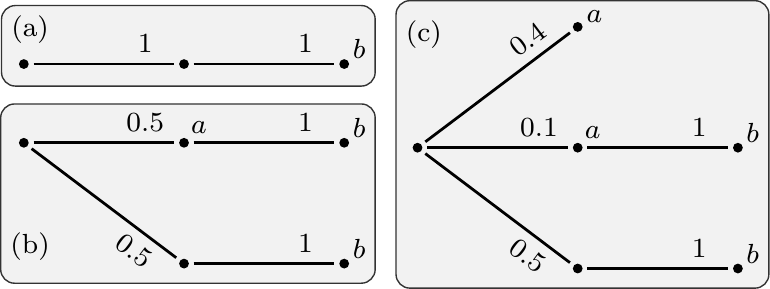}
\caption{Three models of the formula $\phi_0$ from Example~\ref{exa:form}.}
\label{fig:exa0}
\end{figure}
In (a), $a$ and $\bigcirc b$ are observed at the next timepoint with probability 0 and 1, respectively; (c) depicts the
other extreme, where $a$ and $\bigcirc b$ are observed with probability $0.5$ and $0.6$, respectively. 
Many intermediate models exist. 
The formula $\phi_1:=\prob{\ge 0.5}a\land\prob{\ge 0.6}\neg a$ is 
unsatisfiable: no model allows $a$ and $\neg a$ to hold with probability 0.5 and 0.6, 
respectively.
\end{example}
As usual, the main reasoning task for a \PLTL formula $\varphi$ is checking satisfiability. 
This problem becomes harder than in \LTLf by the need to verify the compatibility of the potential future
steps w.r.t.\ their probabilities. Thus, we build a tree automaton accepting a class of models
of $\phi$. For satisfiability, we can ignore the specific probabilities used, as long as they
are compatible. Thus an unweighted automaton suffices. Building this automaton
requires some definitions.

For a \PLTL formula $\phi$, $\csub(\phi)$ is the smallest set of \PLTL formulas containing all
subformulas of $\phi$, is closed under negation (modulo double negations), and s.t.\
$\psi_1 U\psi_2\in\csub(\phi)$ implies $\bigcirc\psi_1 U\psi_2\in\csub(\phi)$.
An \emph{atom} is a subset $\atom\subseteq\csub(\phi)$ s.t.\ (i) for every $\psi\in\csub(\phi)$, 
$\{\psi,\neg\psi\}\cap \atom\not=\emptyset$ and $\{\psi,\neg\psi\}\not\subseteq \atom$; (ii) for every formula
$\psi_1\land\psi_2\in \csub(\phi)$, $\psi_1\land\psi_2\in \atom$ iff $\{\psi_1,\psi_2\}\subseteq \atom$; and 
(iii)~for all $\psi_1 U\psi_2\in\csub(\phi)$, $\psi_1 U\psi_2\in \atom$ iff either $\psi_2\in \atom$ or 
$\bigcirc\psi_1 U\psi_2\in \atom$. Atoms are maximally consistent subsets of $\csub(\phi)$ that 
also verify the satisfiability of the until operator. The set of all atoms is denoted by $\At(\phi)$.
For brevity, we equate $\neg\prob{\bowtie p}\equiv\prob{\bowtie^- p}$, where $\bowtie^-$ is the inverse relation
of $\bowtie$ and assume that probabilistic formulas are never negated in \csub; e.g., $\neg\prob{\ge 0.5}a$
is replaced by $\prob{<0.5}a$.

Atoms define the states of the automaton. To define the transitions, we identify the 
combinations of probabilistic subformulas that can appear together under the uncertainty constraints;
e.g., if an atom contains $\prob{\le 0.3}\psi_1$ and $\prob{\le 0.4}\psi_2$, then transitions
must contain a successor with $\neg\psi_1$ and $\neg\psi_2$. Let 
$\Pmc(\atom)=\{\prob{\bowtie p}\psi\in \atom\}$ be the set of all probabilistic formulas in the atom \atom.
For every subset $S\subseteq 2^{\Pmc(\atom)}$ define the system of inequalities
\begin{align*}
\Imf(S) := & \{ \sum_{\mathclap{\prob{\bowtie p_i}\psi_i\in Q, Q\in S}} x_Q \bowtie p_i\mid \prob{\bowtie p_i}\psi_i\in \Pmc(a)\} \cup {} \\ &
	\{x_Q \ge 0 \mid Q\in S \} \cup \{ \sum_{Q\in S}x_Q = 1\}.
\end{align*}
$\Smc(\atom)$ is the set of all $S\subseteq 2^{\Pmc(\atom)}$ where $\Imf(S)$ has a solution.

\begin{example}
\label{exa:ineq}
One atom of the formula $\phi_0$ from Example~\ref{exa:form} is 
$\atom_0=\{\phi_0,\prob{\le 0.5}a,\prob{\ge 0.6}\bigcirc b,\neg\bigcirc b,\neg a,\neg b\}$; for which
$\Pmc(\atom_0)=\{\prob{\le 0.5}a,\prob{\ge 0.6}\bigcirc b\}$. For brevity, call the elements of $\Pmc(\atom_0)$
$1$ and $2$, respectively. The system of inequalities for $S_0=\{\{1\}, \{2\}, \{1,2\}\}$ 
\begin{align*}
x_{\{1\}} + x_{\{1,2\}} & {} \le 0.5 \\
x_{\{2\}} + x_{\{1,2\}}  & {} \ge 0.6 \\
x_Q & {} \ge 0 & Q\in S_0  \\
x_{\{1\}} + x_{\{2\}} + x_{\{1,2\}} & {} = 1,
\end{align*}
has a solution; e.g., $x_{\{1\}}=0.4$, $x_{\{2\}}=0.5$, $x_{\{1,2\}}=0.1$. As shown later, this means that
from the atom $\atom_0$, it is possible to branch in three scenarios to satisfy the probabilities (Figure~\ref{fig:exa0} (c)).
For $S_1=\{\emptyset,\{1\},\{1,2\}\}$, the system $\Imf(S_1)$
$$ x_{\{1\}} + x_{\{1,2\}} \le 0.5 , x_{\{1,2\}}\ge 0.6, x_\emptyset + x_{\{1\}} + x_{\{1,2\}} = 1$$
has no non-negative solution: $x_{\{1,2\}}$ needs to be $\le 0.5$ and $\ge 0.6$ simultaneously.
Hence $S_0\in\Smc(\atom_0)$ but $S_1\notin\Smc(\atom_0)$.
The latter means that to satisfy $\atom_0$, one must find a transition where the formula $\bigcirc b$ is satisfied, but $a$ is not.
\end{example}
If $\Pmc(\atom)=\emptyset$, then $2^{\Pmc(\atom)}=\{\emptyset\}$, and $\Imf(\{\emptyset\})=\{1{=}x_{\emptyset}\}$, 
which has a trivial solution; i.e.,
if an atom $\atom$ contains no probabilistic subformulas, $\Smc(\atom)$ contains only one element, and the construction reduces to
classical \LTLf.
Each element in $\Smc(\atom)$ defines a set of tuples of atoms yielding the 
transition relation of the automaton. Assume w.l.o.g.\ that the elements
of each $S\in\Smc(\atom)$ are ordered as $Q_1,\ldots,Q_{|S|}$. 
$T_S(\atom)$ is the set of $|S|$\mbox{-}tuples of atoms $(\atom_1,\ldots,\atom_{|S|})$ s.t.\ 
for all $\bigcirc\psi,\prob{\bowtie p}\psi\in\csub(\phi)$: (i) $\bigcirc\psi\in \atom$ iff $\psi\in \atom_i$ for all $i$, and
(ii) for every $i,1{\le} i{\le} |S|$, $\prob{\bowtie p}\psi\in Q_i$ iff $\psi\in \atom_i$.

\begin{example}
From Example~\ref{exa:ineq}, $S_0\in \Smc(\atom)$ defines 3-tuples of atoms 
s.t.\ the first two elements contain one of the probabilistic subformulas each, and the last contains both probabilistic
subformulas. Hence, the tuple $(\atom_1,\atom_2,\atom_3)$ formed by the following atoms belongs to $T_{S_0}(\atom)$:%
\footnote{Recall that we equate $\neg\prob{\le 0.5}\phi\equiv\prob{>0.5}\phi$.}
\begin{align*}
\atom_1 = {} & \{\neg\phi_0, \prob{> 0.5}a,\prob{< 0.6}\bigcirc b,\neg\bigcirc b, a,\neg b\} \\
\atom_2 = {} & \{\neg\phi_0, \prob{> 0.5}a,\prob{< 0.6}\bigcirc b,\bigcirc b,\neg a,\neg b\} \\
\atom_3 = {} & \{\neg\phi_0, \prob{> 0.5}a,\prob{< 0.6}\bigcirc b,\bigcirc b, a,\neg b\}
\end{align*}
\end{example}

We define an automaton which can decide satisfiability of \PLTL formulas.

\begin{definition}
The tree automaton 
$\Amc_\phi=(\Qmc,\Delta,I,F)$ is given by $\Qmc=\At(\phi)$, 
$\Delta=\{\{\atom\}\times\bigcup_{S\in\Smc(\atom)}T_S(\atom)\mid \atom\in \Qmc\}$,
$I=\{\atom\in \Qmc\mid \varphi\in \atom\}$, and
$F$ the set of all atoms not containing formulas of the form $\bigcirc \psi$, $\prob{> p}\psi$, or $\prob{\ge p'}\psi$, $p'> 0$.
\end{definition}
$\Amc_\phi$ naturally generalises the automata-based approach for
satisfiability of \LTLf formulas: if $\phi$ has no probabilistic constructor (i.e., it is an \LTLf formula),
$\Amc_\phi$ is the standard automaton for this setting~\cite{GiMM14}.
$\Amc_\phi$ accepts a class of well-structured quasi-models of the formula $\phi$, merging
redundant branches. The only missing element to have a model are the 
probabilistic values attached to each successor of a node. These are found solving the system of
inequalities built from each transition. 

%\begin{theorem}
\begin{restatable}{theorem}{satempty}
\label{thm:satempty}
The formula $\phi$ is satisfiable iff $\Lmc(\Amc_\phi)\not=\emptyset$.
\end{restatable}
%\end{theorem}

Automata emptiness is decidable in time $\Omc(|\Qmc|^{k+2})$, where
$k$ is the rank of the automaton.
In this case, the rank of $\Amc_\phi$ depends on the size of the formula $\phi$; specifically,
on the number of probabilistic subformulas that it contains: if $\csub(\phi)$ has $n$ probabilistic
subformulas, the rank of $\Amc_\phi$ is bounded by $2^n$\negmedspace; i.e., emptiness of $\Amc_\phi$ runs
in time $\Omc(|\Qmc|^{2^n})$. There is also a non-deterministic algorithm that uses space
$\Omc(|\Qmc|\cdot {2^n})$.
As the number of states is bounded exponentially on the length of $\phi$, Savitch's theorem~\cite{Savi-70}
yields the following result.
\begin{theorem}
\PLTL satisfiability is decidable in exponential space in the number of probabilistic formulas, but \emph{only}
exponential time in the size of the formula.
\end{theorem}
If the total number of probabilistic formulas is bounded by some constant, or if we parameterise the problem over
the number of probabilistic subformulas~\cite{DaFe12}, then satisfiability of \PLTL formulas is in \ET.
Conversely, satisfiability is \ET-hard on the length of the input formula $\phi$. The proof of this 
fact is based on a reduction from the \emph{intersection non-emptiness} problem for deterministic tree automata 
\cite{TATA2007,Seid-94}.

%\begin{theorem}
\begin{restatable}{theorem}{EThard}
\label{thm:EThard}
\PLTL satisfiability is \ET-hard.
\end{restatable}
%\end{theorem}

\section{Probabilistic Entailment}
\label{sec:entailment}

We have focused in a decision problem considering only the existence of a model of the \PLTL formula
$\phi$, disregarding the probabilistic information included in $\phi$. We now
consider reasoning problems dealing with the likelihood of different traces.
A basic probabilistic reasoning problem on \PLTL is computing the \emph{most likely}
trace, along with its probability. To handle the multiplicity of models, we use an
optimistic approach, which selects the model maximising the likelihood of observing a given trace.
One could have chosen a \emph{pessimistic} approach minimising the likelihood instead. Such a case
can be handled analogously by changing all relevant maxima for minima in the following.

\begin{definition}
A \emph{trace} is a finite sequence of propositional valuations. The interpretation $I=(T,\cdot^I,P)$ 
\emph{contains} the trace $t$ if there is a branch $\branch$ of $T$ such that $\branch^I=t$. 
The \emph{probability} of $t$ in $I$ is $P_I(t)=\prod_{w\in \branch}P(w)$.
The \emph{probability} of $t$ w.r.t.\ the \PLTL formula $\phi$ is 
$P_\phi(t)=\max_{I\models\phi}P_I(t)$.
\end{definition}

\begin{example}
Using the formula $\phi_0$ from Example~\ref{exa:form}, let $I_0$ and $I_1$ be the models (b) and (c)
from Figure~\ref{fig:exa0}, respectively. $I_1$ contains the trace $(\emptyset,\{a\})$, but 
$I_0$ does not. Both models contain the trace $t=(\emptyset,\{a\},\{b\})$; $P_{I_0}(t)=0.5$; and $P_{I_1}(t)=0.1$.
$P_{\phi_0}(t)=0.5$ as witnessed by the model $I_0$.
\end{example}
We want to find the traces with a maximal probability. Formally, $t$ is a \emph{most likely trace} (mlt) w.r.t.\ 
$\phi$ iff for every trace $t'$, $P_\phi(t')\le P_\phi(t)$.
In our running example, a most likely trace is $(\emptyset,\emptyset,\{b\})$, which has probability 1 (Figure~\ref{fig:exa0} (a)); 
however, mlts are not necessarily unique; indeed, $(\emptyset,\emptyset,\{a,b\})$ and
$(\emptyset,\emptyset,\{b\},\emptyset)$ are also mlts w.r.t.\ $\phi_0$.
To find the mlts w.r.t.\ $\phi$, we transform the tree automaton $\Amc_\phi$ into
a weighted string automaton $\Bmc_\phi$ which keeps track of the most likely transitions available
from a given state of $\Amc_\phi$. For brevity, we do not distinguish
between the valuation forming a model, and the atom (containing the valuation) of the run of the automaton.
Using the probabilistic semiring, the behaviour of 
$\Bmc_\phi$ yields the probability of the mlts. We later show how to use this
information to extract the actual traces.

Recall that the reduced automaton $\dddot{\Amc_\phi}$ of the emptiness test excludes the bad states from $\Amc_\phi$
and accepts the same language as $\Amc_\phi$, but from every state in $\dddot{\Amc_\phi}$ one can build
a successful run. Deleting bad states also
removes all transitions (produced by the different combinations of the probabilistic subformulas that appear in
an atom) which cannot be used due to semantic incompatibilities. 
Let $(\atom,\atom_1,\ldots,\atom_n)\in\dddot\Delta$; i.e., a transition from 
$\dddot{\Amc_\phi}$. There exists an $S\in\Smc(\atom)$ such that $(\atom_1,\ldots,\atom_n)\in T_S$. For each 
$Q\in S$, we solve the optimisation problem
\begin{align*}
\mathsf{maximize}\ &\  x_Q & \text{subject to }\ & \ \Imf(S).
\end{align*}
Intuitively, we compute the largest probability that a branch satisfying the probabilistic constraints in $Q$ can
obtain, given the other branches defined by $S$. Call the result of this optimisation
problem $\ms_{S,\atom}(Q)$. As each optimisation problem is solved independently, the
maxima may not add 1; e.g., for $\atom_0,S_0$ from
Example~\ref{exa:ineq}, $\ms_{S_0,\atom_0}(\{2\})=1$, $\ms_{S_0,\atom_0}(\{1,2\})=0.5$ and
$\ms_{S_0,\atom_0}(\emptyset)=\ms_{S_0,\atom_0}(\{1\})=0.4$. This is intended; we try to
identify the largest probability that can be assigned to a path in a model; i.e., a trace.

For an atom \atom and a set $Q\subseteq\Pmc(\atom)$, let now 
\[
\ms_\atom(Q):=\max_{S\in\Smc(\atom),Q\in S}\ms_{S,\atom}(Q).
\]
We obtain the automaton $\Bmc_\phi$ by \emph{flattening} $\dddot{\Amc_\phi}$ into a string automaton, and
weighting every transition according to \ms. 

\begin{definition}
\label{def:waphi}
$\Bmc_\phi=(\dddot\Qmc,\ini,\wt,\dddot{F})$ is the weighted automaton 
where $\dddot\Qmc$
and $\dddot{F}$ are obtained from $\dddot{\Amc_\phi}$, 
$\ini(\atom)=1$ iff $\atom\in\dddot{I}$ (0 o.w.), and $\wt(\atom,\atom')=\ms_\atom(Q)$, where $Q\in \atom'$.
\end{definition}
Importantly, $\Bmc_\phi$ is constructed from $\dddot{\Amc_\phi}$ and not from $\Amc_\phi$. This 
ensures that branches defined by unsatisfiable constraints are ignored. 
For example, if $\{\prob{\le p} a,\prob{\le q}\neg a\}\subseteq \atom$, with $p+q<1$,
$\Smc(\atom)\not=\emptyset$, and $\ms_\atom(\{\prob{\le p} a\})=p$. However, if \atom is not a final state,
but a bad state: \atom has no transition. Constructing $\Bmc_\phi$ from $\Amc_\phi$,
\atom would have a transition to an atom $\atom'$ containing $a$ (with weight $p$), which is incorrect.%

\begin{restatable}{theorem}{probmlt}
\label{thm:probmlt}
Let $\Bmc_\phi$ be the weighted automaton constructed from $\phi$ by
Definition~\ref{def:waphi}. The probability of the mlt w.r.t.\ $\phi$ is $\|\Bmc_\phi\|$.
\end{restatable}

The behaviour of $\Bmc_\phi$ is computable in polynomial time on the number of states, i.e., exponential on the length of the formula, but is not affected by the number of probabilistic subformulas
appearing in $\phi$. Yet, to build $\Bmc_\phi$, we need first to construct and manipulate 
$\Amc_\phi$, which may be doubly-exponential on the number of probabilistic subformulas. Indeed,
there is a trace with positive probability iff $\phi$ is satisfiable. Thus, deciding whether the probability
of the most likely trace is higher than some bound has the same complexity as satisfiability.

\begin{corollary}
The probability of the mlts is computable in exponential space in the number of probabilistic 
formulas, but exponential time in the size of the formula. Deciding if it is greater than
0 is \ET-hard.
\end{corollary}

To find the mlts (and not just their probability), we adapt the computation process for 
the behaviour of $\Bmc_\phi$. At each iteration of the  computation, associate each
state with the maximum probability that can be derived from a trace starting from
it. Together with this number, we also store the successor states yielding that maximum probability,
getting an automaton that accepts all the most likely traces.

\begin{definition}
Given a \PLTL formula $\phi$, its weighted automaton $\Bmc_\phi=(\Qmc,\ini,\wt,F)$,
and a state $\atom\in\Qmc$, let $\swt(\atom)$ be obtained through the computation
of the behaviour of $\Bmc_\phi$. The
unweighted automaton $\overline{\Bmc_\phi}=(\Qmc,I,\Delta,F)$ is given by 
\begin{align*}
I = {} & \{\atom\in\Qmc\mid \ini(\atom)\cdot\swt(\atom)=\|\Bmc_\phi\|\}, \\
\Delta = {} &\{(\atom,\atom')\in\Qmc\times\Qmc\mid \wt(\atom,\atom')\cdot\swt(\atom')=\swt(\atom)\}. 
\end{align*}
\end{definition}

%\begin{theorem}
\begin{restatable}{theorem}{mlt}
$\overline{\Bmc_\phi}$ accepts the most likely traces.
\end{restatable}
%\end{theorem}

While it is important to understand the mlts, it is often more useful to compute the likelihood
of observing a specific trace or an element from a set of traces; e.g., to verify that an unwanted outcome is
unlikely. This problem can be reduced to that of computing the most likely traces, as long
as the set of desired traces is a recognisable language.

\begin{definition}
Let $L$ be a recognisable set of finite traces and $\phi$ a \PLTL formula. The \emph{probability} of $L$ 
w.r.t.\ $\phi$ is 
\[
P_\phi(L) = \max_{t\in L} P_\phi(t).
\]
A trace $t\in L$ is a \emph{most likely trace of $L$} w.r.t.\ $\phi$ if it holds that $P_\phi(t)=P_\phi(L)$.
\end{definition}

Since $L$ is recognisable, there exists an automaton $\Amc_L$ that accepts exactly the traces in $L$.
We can obviously see this automaton as a very simple weighted automaton, whose weights
are all in $\{0,1\}$. To find the mlts of $L$ and their corresponding probability, we 
intersect $\Amc_L$ with $\Bmc_\phi$ and $\overline{\Bmc_\phi}$, respectively, and compute 
$\|\Amc_L\cap\Bmc_\phi\|$ and the language accepted by $\Amc_L\cap \overline{\Bmc_\phi}$,
respectively.

Consider now the same probabilistic problems but in relation to an observed prefix. 
Given a sequence $s$ of propositional valuations, we want to analyse
only traces $t$ that extend $s$. 
Formally, if $\phi$ is a \PLTL formula, and $s$ is a finite sequence of propositional valuations, a \emph{most likely 
trace extending $s$} is a trace $t=s\cdot u$ such that for every trace $t'=s\cdot u'$, 
$P_\phi(t')\le P_\phi(t)$. We want to find all the most likely traces extending $s$, and
their probability. Note that the set of all finite words over the alphabet of propositional valuations 
which extend $s$ is recognisable; indeed, a simple concatenation of the universal automaton to the automaton
that accepts only $s$ recognises this language. Hence, our previous results answer this 
question.

\section{The \PLTLz Fragment of \PLTL}
\label{sec:formal}

We have seen that even the basic task of satisfiability is, in the \PLTL case,  \ET-hard on the length of the formula. To mitigate this complexity, we now focus on the fragment of \PLTL where probabilities can only appear as the top-most temporal
constructor of a conjunctive formula. We call this fragment \PLTLz.
Formally, a \PLTLz formula is a finite set of expressions of the form $\prob{\bowtie p}\varphi$, where
$\varphi$ is a classical \LTLf formula, ${\bowtie}\in\{\le,\ge,<,>\}$, and $p\in[0,1]$.%
\footnote{We consider all the standard abbreviations from \LTLf. In particular, $\Diamond\varphi\equiv\top U \varphi$, where
$\top$ stands for any tautology, and $\Box\varphi\equiv \neg\Diamond\neg\varphi$.}
In terms of processes, $\prob{\bowtie p}\varphi$ expresses that the proportion of traces
of the process that satisfy $\varphi$ is $p$. The set of formulas is interpreted as a conjunction of the probabilistic formulas
appearing in it; that is, a \PLTLz formula is a conjunction of probabilistic constraints.
Note that in this restricted setting, the probabilistic constructor \prob{} refers only to the probability of observing a specific
\LTLf formula, without a reference to the next point in time. We preserve the same notation, to keep consistent with the general
logic \PLTL.

This logic is interesting for two reasons. On the one hand, reasoning about \PLTLz falls down to \PS, matching the complexity of the classical \LTLf case (without probabilities). On the other hand, \PLTLz is suited for describing declarative constraints mined from historical log data of business process executions. More specifically, some \PLTLz patterns can be readily mined from log data using existing declarative process discovery techniques. 

\subsection{Reasoning in \PLTLz}
The superposition semantics of \PLTL collapse in \PLTLz to the more standard
multiple-world semantics. To simplify the presentation, we define a \emph{probabilistic interpretation} as a pair 
$\Pmc=(\Imc,P_\Imc)$, where \Imc is a finite set of \LTLf interpretations and $P_\Imc$ is a discrete probability distribution over \Imc. 
Satisfiability of an \LTLf formula by an \LTLf interpretation is defined as usual~\cite{GiVa13}.
The probabilistic interpretation $\Pmc=(\Imc,P_\Imc)$ is a \emph{model} of the \PLTLz formula
$\Phi=\{\prob{\bowtie p_i}\varphi_i\mid 1\le i\le n\}$ iff for every $i,1\le i\le n$ it holds that
\[
P_\Imc(\{I\in\Imc\mid I\models\varphi_i\})\bowtie p_i;
\]
that is, if the probability of all the models of $\varphi_i$ is $\bowtie p_i$.
For example, the formula $\phi_0$ from Example~\ref{exa:form} is equivalent to the \PLTLz formula 
$\Phi_0:=\{\prob{\le 0.5}a, \prob{\ge 0.6}\bigcirc b\}$. Two models of this formula are depicted in Figure~\ref{fig:zeromodel}.
\begin{figure}[tb]
\centering
\includegraphics[width=\columnwidth]{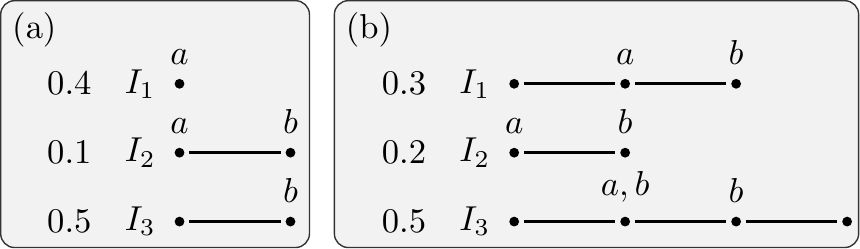}
\caption{A probabilistic model $\Pmc=(\{I_1,I_2,I_3\},P)$ of $\Phi_0$.
The probability of each interpretation appears on the left.}
\label{fig:zeromodel}
\end{figure}

Briefly, the uncertainty of \PLTLz formulas appears only at the beginning of the process, after which the execution follows a 
regular \LTLf execution. Thus, there is no need to branch within the superposition semantics at later times; a model becomes
a degenerate tree which only branches at the root. Moreover, as all formulas start with the constructor \prob{}, the root node serves
only as an anchor for the \PLTL semantics. Hence, a model can be represented as a sequence of classical \LTLf interpretations.
Compare the model in Figure~\ref{fig:exa0} (c) with Figure~\ref{fig:zeromodel} (a).
Interestingly, restricting to \PLTLz reduces the complexity of dealing with probabilistic formulas, and allows for simpler algorithms.
Consider first the case of deciding whether the \PLTLz formula $\Phi$ is satisfiable. 
In practice, this corresponds to verifying whether the class of all possible traces can be divided in such a way that the proportions
required by the probabilistic constraints are satisfied.
To solve this problem, we may proceed as follows.

Given $\Phi=\{\prob{\bowtie p_i}\varphi_i\mid 1\le i\le n\}$, analyse the $2^n$ possible scenarios of a trace, depending which of the
constraints are satisfied or violated. More precisely, consider the $2^n$ sets of constraints in the Cartesian product
$\prod_{i=1}^n\{\phi_i,\neg\phi_i\}$; i.e., each set chooses for every formula whether it will be satisfied or violated.  If each of these sets,
seen as the conjunction of the formulas that it contains, is satisfiable, then the input \PLTLz formula is satisfiable as well. On the other
hand, if any of these sets is unsatisfiable, it means that it is impossible to build a trace that satisfies that combination of formulas;
hence that scenario should be assigned probability 0.

To verify that probabilities for the remaining branches can still be assigned consistently with the values in $\Phi$, we build a system of
inequalities whose solution space is precisely the valid probability assignments. We consider one variable for each case. For readability,
we name these variables $x_0,\ldots, x_{2^n-1}$ using a binary subindex which indicates satisfaction or violation of constraints
assuming w.l.o.g.\ that the constraints are linearly ordered.
That is, the subindex is a chain of length $n$ of 0s and 1s; a 0 or a 1 at position $i$ means that $\neg\phi_i$ or $\phi_i$ is satisfied, 
respectively. We use the same subindex convention to refer to the sets of constraints $S_i$;
e.g., if $\Phi$ contains three formulas, then $x_{010}$ is the variable corresponding to the set 
$S_{010}=\{\neg\phi_1,\phi_2,\neg\phi_3\}$.
Using this information, $\Lmc_\Phi$ is the system of inequalities
\begin{align*}
x_i \ge 0 &&& 0\le i < 2^n \\
\sum_{i=0}^{2^n-1}x_i = 1 \\
\sum_{\text{$j$th position is 1}}x_i \bowtie p_j &&& 0\le j <n\\
x_i = 0 &&& \text{if $S_i$ is unsatisfiable}
\end{align*}
The first two lines guarantee that we assign a non-negative value to each variable, and that their sum is one; we can see these
assignments as probabilities. The third line verifies the probability associated to each constraint in $\Phi$: all the variables that
correspond to cases making $\phi_i$ true should add to be $\bowtie p_i$. The last line ensures that the unsatisfiable cases are never
assigned a positive probability. This system of inequalities has a solution iff the \PLTLz formula is satisfiable.

\begin{example}
\label{exa:run}
Let $\Phi_1:=\{\prob{\le 0.8}\Diamond a, \prob{\le 0.7}\Box(a\to\Diamond b)\}$. Since
$\Phi_1$ has two probabilistic constraints, we build four variables and sets
\begin{align*}
S_{00}:= {} & \{\neg \Diamond a, \neg\Box(a\to\Diamond b)\}, &
S_{01}:= {} & \{\neg \Diamond a, \Box(a\to\Diamond b)\}, \\
S_{10}:= {} & \{\Diamond a, \neg\Box(a\to\Diamond b)\}, &
S_{11}:= {} & \{\Diamond a, \Box(a\to\Diamond b)\}.
\end{align*}
$S_{00}$ is clearly unsatisfiable, but the remaining three sets are satisfiable. 
The system of inequalities must enforce that $x_{00}$ is 0. Specifically, the system $\Lmc_{\Phi_1}$
is 
\begin{align*}
x_{00} =  0 && x_{01} \ge  0 \quad x_{10} \ge 0 \quad x_{11} \ge 0 \\
x_{00} + x_{01} + x_{10} + x_{11} =  1 \\
x_{10} + x_{11} \le  0.8 &&
x_{01} + x_{11} \le  0.7
\end{align*}
A solution of $\Lmc_{\Phi_1}$ is $x_{00}=0$, $x_{01}=0.2$, $x_{10}=0.3$, and $x_{11}=0.5$, which
yields a probabilistic model of $\Phi_1$ consisting
of three interpretations, $I_1, I_2, I_3$; each interpretation $I_i$ satisfies the constraints in the
set $S_i$ and is assigned the probability $x_i$. An interpretation for $S_0$ is not needed because
these constraints are unsatisfiable (and the combination of formulas is assigned probability 0).
\end{example}

\begin{theorem}
The \PLTLz formula $\Phi$ is satisfiable iff $\Lmc_\Phi$ has a solution.
\end{theorem}
To construct $\Lmc_\Phi$, one must solve $2^n$ \LTLf satisfiability tests, each requiring polynomial
space~\cite{SiCl85}. Solving this system of inequalities requires polynomial time on the number of variables; i.e., 
exponential on $n$. Overall, it needs exponential time on $n$, but only polynomial space on the length of $\Phi$.

\begin{theorem}
\PLTLz satisfiability is decidable in exponential time on the number of probabilistic formulas, but in polynomial space on their total
length.
\end{theorem}
In particular, if the number of formulas in $\Phi$ is bounded, satisfiability is \PS-complete, improving the
\ET lower bound for general \PLTL (Theorem~\ref{thm:EThard}).
We are more interested in deducing (probabilistic) guarantees of a process that
follows the constraints in a formula; and more importantly of traces being observed over them. Recall that in our semantics 
any execution is possible as long as there is no evidence to the contrary. In \PLTLz the uncertainty is stated at the beginning; 
that is, we do not know which of the formulas $\phi_1,\ldots,\phi_n$ are satisfied, but once a trace has chosen 
its path, it remains in it without further uncertainty arising later on.

As before, we follow an optimistic approach and try to find the most likely scenarios and traces
that fit the constraints in the formula. Recall that the solution space of $\Lmc_\Phi$ yields
the probability assignments that can be consistently given to the \LTLf interpretations appearing in a probabilistic model
according to the constraints that it satisfies. Thus, maximising the value of a variable $x_i$ yields the maximum probability that
the set $S_i$ may be assigned in a model.

For each $i, 0\le i< 2^n$\negmedspace, let $\max_i$ be the solution of maximising $x_i$ subject to $\Lmc_\Phi$. Note that each
variable is maximised independently of the others, and hence the values $\max_i$ do not form a probability distribution
\emph{per se}; their sum may be greater than 1. The values $\max_i$ express the \emph{optimistic} position
of assigning the highest possible probability to the traces satisfying $S_i$.
For the formula $\Phi_1$ from Example~\ref{exa:run}, the answers to
the maximisation problems for the different variables correspond precisely to the values in the solution presented; namely,
$\max_{00}=0$, $\max_{01}=0.2$, $\max_{10}=0.3$, and $\max_{11}=0.5$.

The question of finding the mlts or simply the \emph{most likely scenario}
can be answered using the values $\max_i$. Take the indices $j$ where $\max_j$
is the largest among all variables; i.e., $j\in\argmax\{\max_i \mid 0{\le} i{<} 2^n\}$.
Each $S_j$ is a most likely scenario, and every trace satisfying $S_j$ is an mlt, with probability $\max_j$. In our
example, $S_{11}$ is the most likely scenario; i.e, we expect to observe a trace satisfying
$\Diamond a$ and $\Box(a\to\Diamond b)$.

If $\Phi$ represents a process, the mlts are those that we would expect to observe in an execution of the process
in the absence of other information. In general, it is more interesting to predict the future behaviour of the process, given
some observation of its first steps. Given a (partial) trace $t$, corresponding to the prefix of a full process, 
we want to find the most likely scenario where $t$ can happen, and predict the potential future evolution of the trace.
Given a set $S$ of \LTLf formulas and a prefix $t$, $S$ \emph{accepts} $t$ (denoted by $S\Vvdash t$)
iff there is a suffix $s$ such that $S\models t\cdot s$.
In words, $S$ accepts a prefix if it can be extended into a trace that satisfies all the conditions in $S$.
To find the most likely scenario accepting a prefix, and a suffix extending it to a successful trace, we generalise the idea 
described for the case without prefix.

Let $\acc(t):=\{i \mid S_i\Vvdash t\}$ be the set of indices $j$ s.t.\ $S_j$ accepts $t$, and let $j\in\argmax_{i\in\acc(t)}\{\max_i\}$
be an index with the maximum value in the set of solutions from the maximisation problems of $\Lmc_\Phi$. Then,
$S_j$ is a most likely scenario given $t$, and any trace extending $t$ accepted by $S_j$ is an mlt. For $\Phi_1$ in
Example~\ref{exa:run}, the most likely scenario for any finite prefix $t$ is always $S_{11}$, and $t\cdot ab$ is always a trace accepted 
by this set. In general, however, the most likely scenario may change as a prefix grows.

\begin{example}
\label{exa:mod}
Let $\Psi_1:=\{\prob{\le 0.5}\Diamond a, \prob{\le 0.6}\Box(a\to\Diamond b)\}$, which is
very similar to $\Phi_1$ (Example~\ref{exa:run}), but with different probabilities. We get
$\max_{00}=0$, $\max_{01}=0.5$, $\max_{10}=0.4$, and $\max_{11}=0.1$. For the empty prefix $\varepsilon$ and the prefix $\neg a$,
the most likely scenario is $S_{01}$, which holds with probability 0.5, and a most likely continuation would append them with a 
finite number of observations of $\neg a$. If at the second point in time we observe $a$ (making the prefix $\neg a\cdot a$) then 
$S_{01}$ does not accept this prefix anymore, and the most likely scenario becomes $S_{10}$: we will eventually 
observe an $a$ after which $b$ will never be observed anymore.
\end{example}
The method for finding the most likely scenario is formalised in Algorithm~\ref{alg:mls},
\begin{algorithm}[tb]
\DontPrintSemicolon
\KwData{$\Phi{=}\{\prob{\bowtie p_i}\varphi_i{\mid} 1{\le} i{\le} n\}$ \PLTLz\negmedspace formula, $t$ prefix}
\KwResult{Index of the most likely scenario for $t$ in $\Phi$}
$mls \gets -1$ \;
\For{$0\le i< 2^n$}{
  compute $\max_i$ \;
  \If{$S_i\Vvdash t$ and $\max_i > \max_{mls}$ }{
    $mls\gets i$
  }
}
{\bf Return} $mls$
\caption{Most likely scenario for $t$ over $\Phi$.}
\label{alg:mls}
\end{algorithm}
where $\max_{-1}:=0$.
Note that an expensive part of this algorithm is the computation of the values $\max_i$. However, this computation can be made
offline, as a preprocessing step before the algorithm is called, as these values remain invariant for any call. A second point to 
consider is that the set $J$ monotonically decreases as the trace $t$ grows. More precisely, for every $t,s$, if $S\Vvdash t\cdot s$,
then $S\Vvdash t$. Hence, if we are monitoring the evolution of a process, trying to find out the most likely continuation of the currently
observed trace, then after every newly observed step, we only need to update the set $J$ to remove those $S_i$s not accepting
the prefix anymore. Finally, to avoid unnecessary tests, we can exclude from the {\bf for} loop all indices $i$ where $\max_i=0$:
$\max_i=0$ means that the system should not observe any trace satisfying $S_i$. If the most likely scenario is one that has
probability 0, then the observed prefix is violating the conditions described by $\Phi$.

\begin{proposition}
Algorithm~\ref{alg:mls} returns the index of the most likely scenario, given a prefix.
\end{proposition}
Interestingly, assuming that all the values $\max_i$ have been computed before, Algorithm~\ref{alg:mls} can be executed to
use only polynomial space. The information to control the {\bf for} loop requires at most $n$ bits, and the two tests within this loop
require polynomial space.

Note that finding the probability of the most likely
scenario (and trace) is akin to monitoring agreement with a model. Analogously, one can extend the
task to monitoring a complex \PLTL property $\psi$. For the maximum likelihood of accepting $\psi$,
we use Algorithm~\ref{alg:mls}, but now considering whether $S_i\cup\{\psi\}\Vvdash t$; i.e.,
finding the scenarios where $\psi$ may still be satisfied given the knowledge of $t$. This is analogous to the notion of
\emph{eventual satisfaction} in monitoring. Other notions like \emph{current satisfaction}, or \emph{permanent satisfaction}
can be dealt with accordingly, modifying the notion of acceptability of a trace w.r.t.\ a set of \LTLf constraints~\cite{DDGM14}.

\subsection{Discovering \PLTLz Patterns from Event Log Data}

\PLTLz formulas can be automatically mined from event log data using state-of-the-art \emph{declarative process discovery} algorithms within \emph{process mining}. Process mining focuses on the continuous improvement of business processes based on factual data \cite{Aal16}. Such data are stored in a so-called \emph{event log}, where each event refers to an
\emph{activity} (a well-defined step in a process) and is related to a \emph{case} (a
\emph{process instance}). Events in a case are \emph{ordered} and seen as an execution (or \emph{trace}) of the process. 
A core process mining task is \emph{process discovery}, which learns a process model that reproduces the traces
contained in the log. In declarative process discovery, the target model is specified using 
rules/constraints, like the \LTLf patterns adopted by the \declare process modelling language \cite{PeSV07}.

\citeauthor{MaCV12} \shortcite{MaCV12} developed a two-phase method to automatically infer Declare constraints from
event logs. In the first phase, candidate constraints to be mined are generated by an algorithm called Apriori.
This algorithm returns
frequent activity sets indicating a high correlation between activities involved in an activity set. Highly correlated sets are used to instantiate, in any possible ways, the \declare patterns. For example, considering the frequent activity set $\{ a, b \}$ and the \declare pattern of \constraint{response}, the two \LTLf constraints $\square (a \rightarrow \lozenge b)$
and $\square (b \rightarrow \lozenge a)$ are generated.  In the second phase, the set of so-generated constraints is filtered by retaining only ``relevant'' constraints, where relevance is measured using metrics such as that of \emph{support}:
the proportion of traces satisfying the constraint.

What makes Apriori interesing in our setting is that support can be interpreted as the constraint probability: the discovery of $\LTLf$ 
formula $\varphi$ with support $p$ (that is, appearing in $100p\%$ of the traces) can be interpreted as the discovery of the $\PLTLz$ 
formula $\prob{= p}\varphi$. 

\section{Conclusions}

We have introduced a new probabilistic temporal logic \PLTL based on a novel superposition semantics, and its sublogic \PLTLz where
this semantics collapses to the standard multiple-world approach. These logics are specifically crafted for predicating about uncertainty
in dynamic systems whose executions eventually finish. We studied the main properties of the logics, and provided automata-theoretic methods for extracting
relevant information from them.

In future work, we plan to implement the algorithms for \PLTLz and apply them to the declarative modelling and analysis of business processes, considering in particular monitoring and conformance checking.

\bibliographystyle{aaai}
\bibliography{main-bib}

\clearpage 

\appendix

\section{Appendix A: Proofs}
%\label{app:proofs}

\satempty*
\begin{proof}{}
[$\Leftarrow$] Suppose first that $\Lmc(\Amc)\not=\emptyset$, and let $T\in \Lmc(\Amc)$. This means that there exists a successful run $\rho$ of
\Amc over $T$, which maps every node $w\in T$ with an atom $\rho(w)$. Let $A$ be the set of all propositional variables appearing
in $\phi$. Define the function $\cdot^I: T\to 2^A$ by $w^I:=\rho(w)\cap A$. Moreover, for every node $w$ with $k$ successors, 
as $\rho$ is a successful run, it holds that $(\rho(w1),\ldots,\rho(wk))\in T_S$ for some $S\in\Smc(\rho(w))$. The latter means that 
the system $\Imf(S)$ has a solution for the (ordered) variables $x_1,\ldots,x_k$ in $[0,1]$. Hence, we define the function 
$P:T\setminus\{\varepsilon\}\to [0,1]$ where $P(w\ell)$ is the solution of the system in $\rho(w)$ for the variable $x_\ell$. 
Overall, this defines an interpretation $I=(T,\cdot^I,P)$. 
We show by induction
on the structure of the formulas that for every $\psi\in\csub(\phi)$ and every $w\in T$, $I,w\models \psi$ iff $\psi\in\rho(w)$.
In particular, since $\rho$ is such that $\phi\in\rho(\varepsilon)$, this implies that $I$ is a model of $\phi$.

For a propositional variable $a\in A$ the result holds trivially by construction, so we focus on the remaining constructors. Assume
that the result holds for every formula in $\csub(\psi)\cup\csub(\psi_1)\cup\csub(\psi_2)$.

\noindent[$\neg$] $I,w\models\neg\psi$ iff $I,w\not\models\psi$ iff (induction hypothesis) $\psi\notin\rho(w)$ iff (atom maximality)
	$\neg\psi\in\rho(w)$.
	
\noindent[$\land$] $I,w\models\psi_1\land\psi_2$ iff $I,w\models\psi_1$ and $I,W\models\psi_2$ iff (induction hypothesis)
	$\{\psi_1,\psi_2\}\subseteq\rho(w)$ iff (atom condition (ii)) $\psi_1\land\psi_2\in\rho(w)$.
	
\noindent[$\bigcirc$] $I,w\models\bigcirc\psi$ iff $w$ is not a leaf and for every $wi\in T$ $I,wi\models\psi$ iff for every $wi\in T$, 
	$\psi\in\rho(wi)$ iff (definition of the transition relation) $\bigcirc\psi\in\rho(w)$.
	
\noindent[$U$] $I,w\models \psi_1 U \psi_2$ iff (i) $I,w\models\psi_2$ or (ii) $I,w\models\psi_1$ and for all $wi\in T$, 
	$I,wi\models \psi_1 U\psi_2$. We show that $\psi_1 U\psi_2\in\rho(w)$ by induction on the subtree rooted at $w$. 
	If $w$ is a leaf node, only case (i) is possible, and so $I,w\models \psi_1 U \psi_2$ iff $\psi_2\in \rho(w)$ which means that 
	$\psi_1 U \psi_2\in \rho(w)$ by the definition of an atom.
	If $w$ is not a leaft node. Case (i) is treated as for the leaf nodes. Case (ii) holds iff $\psi_1\in\rho(w)$ and, by the second
	induction, for every $wi\in T$, $\psi_1 U\psi_2\in\rho(wi)$, which implies by the definition of the transition relation that 
	$\bigcirc(\psi_1 U \psi_2)\in\rho(w)$, and since $\rho(w)$ is an atom, $\psi_1 U\psi_2\in\rho(w)$.
	
\noindent[$\prob{}$] $I,w\models\prob{\bowtie p}\psi$ iff $\sum_{wi\in T,I,wi\models\psi} P(wi)\bowtie p$. By construction and the
induction hypothesis, the latter holds iff $\sum_{wi\in T,\psi\in\rho(wi)}x_i\bowtie p$, where the $x_i$s for the solution of the system in
$\rho(w)$. But this is only possible if $\prob{\bowtie p}\psi\in Q_i\subseteq\rho(w)$.

Hence, $I$ is a model of $\phi$ and $\phi$ is satisfiable. This finishes this direction of the proof.

\noindent[$\Rightarrow$] Conversely, suppose that $\phi$ is satisfiable, and let $I=(T,\cdot^I,P)$ be a model of $\phi$. We will
use this tree to construct a successful run of \Amc, but it needs to be adapted to a simplified form. Given a node $w\in T$, let 
$\Pmc(w)\subseteq\csub(\phi)$ be the set of all probabilistic formulas $\prob{\bowtie p}\psi\in\csub(\phi)$ such that 
$I,w\models \prob{\bowtie p}\psi$, and define $S(w)\subseteq 2^{\Pmc(w)}$ to be the set of subsets $O\subseteq \Pmc(w)$
such that there is a successor $wi\in T$ that satisfies $I,wi\models \psi$ for all $\prob{\bowtie p}\psi\in O$ and
$I,wi\not\models \psi$ for all $\prob{\bowtie p}\psi\notin O$. Given an $O\in S(w)$, if there are two successors $wi, wj\in T$ that 
satisfy the previous conditions, it is possible to \emph{prune} the tree $T$ by removing all nodes of the form $wjv, v\in \nats^*$,
and setting $P'(wi):=P(wi)+P(wj)$. It is easy to see that $I'=(T',\cdot^{I'},P')$, where $T'$ is the prunned tree and $\cdot^{I'}$ is 
$\cdot^I$ restricted to $T'$ is also a model of $\phi$. Let $I_0=(T_0,\cdot^{I_0},P_0)$ be the result of applying this prunning procedure
to all nodes in the original model (in a top-down manner). Then, it is a simple exercise to verify that the labelling
$\rho:T_0\to\At(\phi)$ where $$\rho(w)=\{\psi\in\csub(\phi)\mid I_0,w\models\psi\}$$
is in fact a successful run of \Amc, and hence $\Lmc(\Amc)\not=\emptyset$.
\end{proof}

\EThard*
\begin{proof}
We prove \ET-hardness by a reduction from the intersection non-emptiness problem for deterministic automata over labelled trees.
A \emph{deterministic tree automaton over labelled trees} is a tuple $\Amc=(\Qmc,\Sigma,\Delta,I,F)$ where \Qmc, $I$, and $F$
are as in the preliminaries, $\Sigma$ is a finite \emph{alphabet}, and $\Delta:\Qmc\times\Sigma\to \bigcup_{i\le k}\Qmc^k$
is a total \emph{transition function}. %, where the number of successors is given by the rank of the symbol $\sigma$. 
Given a $\Sigma$-labelled tree, a run of this automaton is a \Qmc-labelling that is consistent
with the transition function. All the associated notions are defined in the obvious way. The intersection non-emptiness problem
for these automata consists in deciding, given $n$ such automata $\Amc_i, 1\le i\le n$ with disjoint sets of states, whether 
$\bigcap_{i=1}^n \Lmc(\Amc_i)\not=\emptyset$. This problem is known to be \ET-hard \cite{Seid-94}.

Given a deterministic automaton $\Amc=(\Qmc,\Sigma,\Delta,I,F)$, we build the \PLTL formula $\varphi_\Amc$ as follows. 
The propositional variables appearing in the formula will be given by the elements of $\Qmc\cup\Sigma\cup\{1,\ldots,k\}$; that is, the 
states, the symbols, and the first $k$ natural numbers. Intuitively, an interpretation and node mapping $q\in \Qmc$ to means that the state 
$q$ holds in that element, and analogously for $\sigma\in\Sigma$. The variables $1,\ldots, k$ are used to distinguish the different
successors of a node in a tree. This intuition will become more clear after the construction of the formula. In addition, we have
a new propositional variable $x_\END$ that identifies the leaf nodes.

For every $(q,\sigma)\in\Qmc\times\Sigma$, let $\Delta(q,\sigma)=(q^{q,\sigma}_1,\ldots,q^{q,\sigma}_{k_{q,\sigma}})$ 
and define the formulas $\psi_{q,\sigma}$, $\psi_Q$, $\psi_\Sigma$, and $\psi_N$ as in Figure~\ref{fig:fml}.
\begin{figure*}[t!]
\begin{align*}
\psi_{q,\sigma} := {} & 
	\begin{cases}
		q \land \sigma \to \left(x_\END \lor \bigwedge_{i=1}^{k_{q,\sigma}} \prob{\ge 1/k_{q,\sigma}}(q^{q,\sigma}_i \land i) \right) & q\in F \\
		q \land \sigma \to \left(\neg x_\END \land \bigwedge_{i=1}^{k_{q,\sigma}} \prob{\ge 1/k_{q,\sigma}}(q^{q,\sigma}_i \land i) \right) & q\notin F
	\end{cases} \\
%\]
%\[
\psi_Q := {} & \bigvee_{q\in \Qmc} \left( q \land \bigwedge_{q'\in \Qmc\setminus\{q\}} \neg q' \right) \\
\psi_\Sigma := {} & \bigvee_{\sigma\in \Sigma} \left( \sigma \land \bigwedge_{\sigma'\in \Sigma\setminus\{\sigma\}} \neg \sigma' \right) \\
\psi_N := {} & \bigwedge_{i=1}^k \left(i \to \bigwedge_{j\not=i} \neg j \right)
\end{align*}
\caption{Formulas describing the automaton for the proof of Theorem~\ref{thm:EThard}.}
\label{fig:fml}
\end{figure*}
Then, we set
\[
\varphi_\Amc := \bigvee_{q\in I}q\land
	\Box\left( \psi_Q \land \psi_\Sigma \land \bigwedge_{(q,\sigma)\in \Qmc\times\Sigma}\psi_{q,\sigma}\right).
\]
Note that the length of $\varphi_\Amc$ is polynomially bounded by the size of \Amc.
We first show that every model of $\varphi_\Amc$ can be transformed into a tree accepted by \Amc, and conversely,
every tree accepted by \Amc, together with a successful run, is a representation of a model of $\varphi_\Amc$. 

Let $J=(T,\cdot^J,P)$ be a model of $\varphi$. By construction, (see formulas $\psi_Q$ and $\psi_\Sigma$), for every node $w\in T$ there 
is exactly one $\sigma\in \Sigma$ and one $q\in \Qmc$ such that $\sigma,q\in w^J$. Abusing the notation, we call these elements
$\Sigma(w)$ and $Q(w)$, respectively. We can assume, w.l.o.g., that for every non-leaf
node $w\in T$, if $q,\sigma\in w^J$, then for every $j,1\le j\le k_{q,\sigma}$, $wj\in T$ and $j\in wj^J$.%
\footnote{Otherwise, one can merge different successors $w\ell$ such that $j\in w\ell^J$ and reorder the successors to satisfy the 
condition.}
We construct the labelled tree $T_J: T\to \Sigma$ where $T_J(w)=\Sigma(w)$ for all $w\in T$. We show that $T_J\in \Lmc(\Amc)$, by
using the function $Q:T\to \Qmc$ to build a successful run of \Amc on this tree. Note that for the root node $\varepsilon$, 
$Q(\varepsilon)\in I$, because $J$ is a model of $\bigvee_{q\in I}q$. For every leaf node $w\in T$, $Q(w)\in F$ because otherwise
the formula $\psi_{q,\sigma}$ guarantees that $w$ must have a successor node. Finally, given a non-leaf node $w\in T$, let 
$q=Q(w)$ and $\sigma=\Sigma(w)$ and $\Delta(q,\sigma)=(q^{q,\sigma}_1,\ldots,q^{q,\sigma}_{k_{q,\sigma}})$. Since $J$ is a model 
of $\psi_{q,\sigma}$, and by the assumption stated before, $w$ must have $k_{q,\sigma}$ successors, and is such that
$Q(wj)=q_j^{q,\sigma}$ for all $j,1\le j\le k_{q,\sigma}$. Thus, the labelling of $T$ provided by the function $Q$ forms a successful run,
and $T_J\in \Lmc(\Amc)$.

For the second claim, let $T$ be a $\Sigma$-labelled tree, and $Q:T\to\Qmc$ a successful run of \Amc over $T$. We build a model
$J=(T,\cdot^J,P)$ of $\varphi_\Amc$ as follows. For every $w\in T$, $\sigma\in\Sigma$, and $q\in \Qmc$, we have 
$\sigma\in w^J$ iff $T(w)=\sigma$ and $q\in w^J$ iff $Q(w)=q$. In addition, for every node $wj\in T$, $j\in wj^J$, and for every
leaf node $w$ $x_\END\in w^J$. Finally, if $w$ has $\ell$ successors, then for every $i,1\le i\le \ell$, $P(wi)=1/\ell$. It is easy to see
that, since $Q$ is a successful run of \Amc over $T$, this interpretation satisfies all the constraints of $\varphi_\Amc$, and hence it
is a model of this formula.

Given $n$ deterministic tree automata $\Amc_1,\ldots,\Amc_n$, %where $\Amc_i=(\Qmc_i,\Sigma,\Delta_i,I_i,F_i)$
we construct the \PLTL formula $\varphi:=\bigwedge_{i=1}^n\varphi_{\Amc_i}$,
whose length is polynomially bounded by the total length of the $n$ automata.  We show, making use of the previous arguments,
that this formula is satisfiable iff the intersection of these automata is not empty.

If there is a tree $T$ in the intersection of these automata, then there is a successful run for $\Amc_i$ over $T$ for all 
$i,1\le i\le n$. We can use these successful runs to build a model of $\varphi$ as did in the previous paragraph. Conversely,
let $J$ be a model of $\varphi$. In particular, $J=(T,\cdot^J,P)$ is a model of $\varphi_{\Amc_i}$ for all $i,1\le i\le n$. Thus, as we have 
shown already, the labelled tree $T_J$ (which only depends on the alphabet symbols $\Sigma$) is in $\Lmc(\Amc_i)$ for all
$i$; i.e., $T_J\in \bigcap_{i=1}^n \Lmc(\Amc_i)$, and hence the intersection is not empty.
\end{proof}

\probmlt*
\begin{proof}
Notice first that the probability of the mlt is zero iff $\phi$ is unsatisfiable. In this case, the automata $\dddot{\Amc_\Phi}$ and 
$\Bmc_\phi$ become empty, and hence the behaviour of the latter is zero as well. So we are only interested in cases where this
probability is greater than 0.

Consider first a run $\rho=q_0,\ldots,q_n$ of $\Bmc_\phi$ such that $\wt(\rho)>0$. In particular this means that 
$\wt(q_i,q_{i+1})>0$ for all $i,1\le i<n$. By construction, this means that for every $i,1\le i<n$ there is a transition of $\dddot{\Amc_\phi}$
(i.e., a tuple $\delta\in\dddot\Delta$) of the form $(q_i,\ldots,q_{i+1},\ldots,q_k)$ such that $\wt(q_i,q_{i+1})$ is the maximum 
value that can be given to a successor of a node satisfying $q_i$ which satisfies $q_{i+1}$. Moreover, all states appearing in this 
transition are \emph{good} states, which means that a successful run can still be constructed from them. Thus, there is a 
successful run of $\dddot{\Amc_\phi}$ (and hence of $\Amc_\phi$) which has $\rho$ as a branch. If $\ini(q_0)=1$ (that is,
if $q_0$ is an initial state of $\Amc_\phi$), as in the proof of Theorem~\ref{thm:satempty}, we can build a model of $\phi$ containing 
a branch $t=\rho(q_0)\cap A,\ldots,\rho(q_n)\cap A$, where $A$ is the set of all propositional variables in $\phi$. Moreover, the
probability of this trace in this model is exactly $\wt(\rho)$. Thus, to summarise, for every successful run $\rho$ of $\Bmc_\phi$, there 
is a model $I$ and a trace $t$ such that $\wt(\rho)=P_I(t)$. This implies that the probability of the mlt is greater or equal to
$\|\Bmc_\phi\|$.

Conversely, consider a model $I$ containing the trace $t$. As in the proof of Theorem~\ref{thm:satempty}, we can assume w.l.o.g.\
that this model translates into a successful run $\rho$ of $\dddot{\Amc_\phi}$. In this model, it holds that for every non-root node $wi$,
$P(wi)\le \ms_{S,\rho(w)}(Q)$, where $S$ and $Q$ are the ones obtained from the transition used in $\rho$. In particular,
$P(wi)\le \ms_{\rho(w)}(Q)$, and hence $P_I(t)\le \|\Bmc_\phi\|$. Since this is true for all traces and all models, it follows that
the probability of the most likely trace is at most $\|\Bmc_\phi\|$.

Piecing both parts together yields the desired result.
\end{proof}

\mlt*
\begin{proof}
By construction, the initial states of $\overline{\Bmc_\phi}$ are exactly those that maximise the likelihood of the trace, and likewise
transitions are only allowed when they preserve the maximum possible probability. That is, for every successful run $\rho$
of $\overline{\Bmc_\phi}$, if seen over the weighted automaton $\Bmc_\phi$ we get that $\wt(\rho)=\|\Bmc_\phi\|$. Following
the arguments from the proof of Theorem~\ref{thm:probmlt}, such a run corresponds to a trace $t$ in a model $I$ such that
$P_I(t)=\|\Bmc_\phi\|$, but since the latter is the probability of the most likely traces, $t$ must be an mlt as well.
\end{proof}

%%%%%%%%%%%%

%\newpage
\section{Appendix B: Example}
\label{app:example}

We now provide a fully developed example of the methods for reasoning with \PLTL, over a slightly more complex formula.
Consider the formula
\[
\psi := \bigcirc\neg b \land \prob{\le 0.7} (a U b) \land \prob{\le 0.6}\bigcirc(\neg a\land \neg b),
\]
which is satisfiable (see Figure~\ref{fig:appmodel} for a model).

The set $\csub(\psi)$ is%
\footnote{We slightly simplify removing three irrelevant conjunctions, and already use the equivalence 
$\neg\prob{\le p}\phi\equiv \prob{>p}\phi$.}
\begin{align*}
\csub(\psi) := \{ & \psi, \neg\psi, \bigcirc\neg b, \neg\bigcirc\neg b, \neg b, b, \\
		& \prob{\le 0.7} (a U b), \prob{> 0.7} (a U b), a U b, \neg (a U b), \\
		& \bigcirc (a U b), \neg \bigcirc(a U b), a, \neg a, \\
		& \prob{\le 0.6}\bigcirc(\neg a\land \neg b), \prob{> 0.6}\bigcirc(\neg a\land \neg b), \\
		& \bigcirc(\neg a\land \neg b), \neg\bigcirc(\neg a\land \neg b), \\
		& \neg a\land \neg b, \neg (\neg a\land \neg b) &\}.
\end{align*}
From this set, we need to construct the class of atoms. As is the case already for \LTLf, this class contains exponentially many
elements on the length of $\psi$, and enumerating them all is not very informative. So we present only a few relevant atoms that
will be useful for highlighting the remaining properties. These are shown in Figure~\ref{fig:atoms}.
\begin{figure*}[tbh]
\begin{align*}
\atom_1 := {} & \{ \psi, \bigcirc\neg b, \prob{\le 0.7}(a U b),\prob{\le 0.6}\bigcirc(\neg a\land\neg b), \neg(a U b), \neg a, \neg b,
			\neg\bigcirc(a U b), \neg a\land \neg b, \neg\bigcirc(\neg a\land \neg b) \} \\
\atom_2 := {} & \{ \psi, \bigcirc\neg b, \prob{\le 0.7}(a U b),\prob{\le 0.6}\bigcirc(\neg a\land\neg b), \neg(a U b), \neg a, \neg b,
			\neg\bigcirc(a U b), \neg a\land \neg b, \bigcirc(\neg a\land \neg b) \} \\[2mm] %\hline \\[-3mm]
\atom_3 := {} & \{ \neg\psi, \neg\bigcirc\neg b, \prob{\le 0.7}(a U b),\prob{\le 0.6}\bigcirc(\neg a\land\neg b), a U b, a, \neg b,
			\bigcirc(a U b), \neg(\neg a\land \neg b), \neg\bigcirc(\neg a\land \neg b) \} \\
\atom_4 := {} & \{ \neg\psi, \neg\bigcirc\neg b, \prob{\le 0.7}(a U b),\prob{\le 0.6}\bigcirc(\neg a\land\neg b), \neg(a U b), a, \neg b,
			\neg\bigcirc(a U b), \neg(\neg a\land \neg b), \bigcirc(\neg a\land \neg b) \} \\
\atom_5 := {} & \{ \neg\psi, \neg\bigcirc\neg b, \prob{\le 0.7}(a U b),\prob{\le 0.6}\bigcirc(\neg a\land\neg b), a U b, a, \neg b,
			\bigcirc(a U b), \neg(\neg a\land \neg b), \bigcirc(\neg a\land \neg b) \} \\[2mm]
\atom_6 := {} & \{ \neg\psi, \neg\bigcirc\neg b, \prob{\le 0.7}(a U b),\prob{\le 0.6}\bigcirc(\neg a\land\neg b), \neg(a U b), \neg a, \neg b,
			\neg\bigcirc(a U b), \neg a\land \neg b, \neg\bigcirc(\neg a\land \neg b) \} \\
\atom_7 := {} & \{ \neg\psi, \neg\bigcirc\neg b, \prob{\le 0.7}(a U b),\prob{\le 0.6}\bigcirc(\neg a\land\neg b), \neg(a U b), a, \neg b,
			\neg\bigcirc(a U b), \neg(\neg a\land \neg b), \neg\bigcirc(\neg a\land \neg b) \} \\
\atom_8 := {} & \{ \neg\psi, \bigcirc\neg b, \prob{> 0.7}(a U b),\prob{\le 0.6}\bigcirc(\neg a\land\neg b), a U b, a, \neg b,
			\bigcirc(a U b), \neg(\neg a\land \neg b), \neg\bigcirc(\neg a\land \neg b) \} \\
\atom_9 := {} & \{ \neg\psi, \neg\bigcirc\neg b, \prob{> 0.7}(a U b),\prob{\le 0.6}\bigcirc(\neg a\land\neg b), a U b, a, \neg b,
			\bigcirc(a U b), \neg(\neg a\land \neg b), \neg\bigcirc(\neg a\land \neg b) \} \\
\atom_{10} := {} & \{ \neg\psi, \neg\bigcirc\neg b, \prob{\le 0.7}(a U b),\prob{\le 0.6}\bigcirc(\neg a\land\neg b), a U b, a, b,
			\neg\bigcirc(a U b), \neg(\neg a\land \neg b), \neg\bigcirc(\neg a\land \neg b) \} 
\end{align*}
\caption{Some atoms for the formula $\phi$.}
\label{fig:atoms}
\end{figure*}

Note that $\atom_1$ and $\atom_2$ differ only on the last element, and in particular contain the same probabilistic formulas, hence 
$\Pmc(\atom_1)=\Pmc(\atom_2)=\{\prob{\le 0.7}(a U b),\prob{\le 0.6}\bigcirc(\neg a\land\neg b)\}$. Define the elements of 
$2^{\Pmc(\atom_1)}$ to be 
\begin{align*}
Q_{00} := {} & \emptyset, &
Q_{01} := {} & \{\prob{\le 0.7}(a U b)\}, \\
Q_{10} := {} & \{\prob{\le 0.6}\bigcirc(\neg a\land\neg b)\}, &
Q_{11} := {} & \Pmc(\atom_1).
\end{align*}
There are in total eight subsets of $2^{\Pmc(\atom_1)}$; in particular, we consider 
$S_0=\{Q_{01},Q_{10},Q_{11}\}$, 
$S_1=\{Q_{01},Q_{10}\}$, and 
$S_2=\{Q_{01},Q_{11}\}$.
Each of these sets defines a system of inequalities, which we now analyse in detail.

Consider first $\Imf(S_0)$, which is defined by%
\footnote{To improve readability, we abuse the notation and use $x_i$ to represent the variable $x_{Q_i}$.}
\begin{align*}
x_{01} \ge {} & 0 & x_{10} \ge {} & 0 & x_{11} \ge {} & 0 \\
x_{01} + x_{10} + x_{11} = {} & 1 \\
x_{01} + x_{11} \le {} & 0.7 \\
x_{10} + x_{11} \le {} & 0.6
\end{align*}
This system is satisfiable; for instance, one solution is given by $x_{01}=0.4$ and $x_{10}=x_{11}=0.3$.

For $S_1$, we obtain the system $\Imf(S_1)$
\begin{align*}
x_{01} \ge {} & 0 & x_{10} \ge {} & 0 \\ %& x_{11} \ge {} & 0 \\
x_{01} + x_{10} = {} & 1 \\
x_{01}  \le {} & 0.7 \\
x_{10}  \le {} & 0.6
\end{align*}
which is also satisfiable; e.g.\ $x_{01}=x_{10}=0.5$.

Finally, consider the system $\Imf(S_2)$
\begin{align*}
x_{01} \ge {} & 0 & x_{11} \ge {} & 0 \\
x_{01} + x_{11} = {} & 1 \\
x_{01} + x_{11} \le {} & 0.7 \\
x_{11} \le {} & 0.6
\end{align*}
Clearly, this system is unsatisfiable because the constraints 
$x_{01} + x_{11} = 1$ and $x_{01} + x_{11} \le 0.7$ are in conflict with each other. 
Overall, this means that $\Smc(\atom_1)$ contains $S_0$ and $S_1$ but not $S_2$; and since
$\Pmc(\atom_1)=\Pmc(\atom_2)$, the same is true for $\Smc(\atom_2)$. Intuitively, what this means is that from
a node satisfying the formulas in $\Pmc(\atom_1)$ (namely, $\prob{\le 0.7}(a U b)$ and $\prob{\le 0.6}\bigcirc(\neg a\land\neg b)$)
there are models that have successors satisfying either the first formula alone, or the second formula alone, but not both, nor none
($S_1$), but there are no models that have successors satisfying only the first formula, or both formulas, but not only the second
nor none ($S_2$). This conclusion is based on the probabilistic constraints alone; the logical interpretation of these formulas
is considered within the automata construction as in classical \LTLf.

Ordering the sets $Q_i$ in the standard numerical order over \nats, we can see that $(\atom_3,\atom_4,\atom_5)\in T_{S_0}(\atom_1)$.
Indeed, $\atom_1$ contains the formula $\bigcirc\neg b$ and $\neg b\in \atom_i, 3\le i\le 5$. Also, $\atom_1$ does not contain
$\bigcirc(a U b)$ nor $\bigcirc(\neg a\land \neg b)$ and we observe that $a U b, \neg a\land\neg b\notin \atom_4$. Regarding the 
probabilistic formulas (i.e., condition (ii) of the definition
of $T_S$) we see that $a U b\in \atom_3\cap \atom_5$ (that is, in the first and third position of the tuples), and 
$\prob{\le 0.7}(a U b)\in Q_{01}\cap Q_{11}$ (that is, the first and third elements of $S_0$), and similarly for
$\prob{\le 0.6}\bigcirc(\neg a\land \neg b)$. Using analogous arguments, we can show that 
$$\{(\atom_3,\atom_2,\atom_5),(\atom_8,\atom_4,\atom_5),(\atom_8,\atom_2,\atom_5)\}\subseteq T_{S_0}(\atom_1).$$

Note, however, that $(\atom_3,\atom_4,\atom_5)\notin T_{S_0}(\atom_2)$. The reason for this is that 
$\bigcirc(\neg a\land \neg b)\in\atom_2$ but as we have said before $\neg a\land \neg b\notin\atom_4$, which violates the first
condition in the definition of $T_S$. Similarly, we can see that $(\atom_3,\atom_4)\in T_{S_1}(\atom_1)\setminus T_{S_1}(\atom_2)$.
In fact, for any $S\in\Smc(\atom_2)$ such that $Q_{01}\in S$ it holds that $T_S(\atom_2)=\emptyset$. To see this, notice that
since $\bigcirc(\neg a\land \neg b)\in \atom_2$, any element in a tuple in $T_S(\atom_2)$ must contain $\neg a\land\neg b$. In particular
this is true for the tuple that corresponds to the position of $Q_{01}=\{\prob{\le 0.7}(a U b)\}$. But then, that element must contain
the set $\{a U b, \neg a, \neg b\}$, which contradicts the definition of an atom. We can similarly show that any $S$ containing $Q_{11}$
defines an empty transition set for $\atom_2$.

Consider now the atom $\atom_5$ and notice that $\Pmc(\atom_5)=\Pmc(\atom_1)$ as well, and hence 
$\Smc(\atom_5)=\Smc(\atom_1)$. We may try to analyse the possible transitions from $\atom_5$ as we did before. However,
notice that $\atom_5$ contains the two formulas $\bigcirc(a U b)$ and $\bigcirc(\neg a\land \neg b)$. This means that for every
$S\in\Smc(\atom_5)$ and every atom \atom appearing in a tuple from $T_S(\atom_5)$ it must hold that 
$\{a U b, \neg a\land \neg b\}\subseteq\atom$, but this contradicts the definition of an atom. Thus, 
$\bigcup_{S\in\Smc(\atom_5)}T_S(\atom_5)=\emptyset$. Note that this is not specific of $\atom_5$, but in fact any atom $\atom$ such
that $(\_,\_,\atom)\in T_{S_0}(\atom_1)$ behaves like this. To see why, note that any such atom must contain the formulas
$a U b$ and $\bigcirc(\neg a\land\neg b)$ arising from the fact that $\atom$ is the successor representing the set $Q_{11}$ and
witnesses the satisfiability of both probabilistic constraints. Moreover, \atom being a successor of $\atom_1$ means that 
$\neg b\in\atom$. Thus, by the properties of atoms, $\bigcirc\neg b$ must belong to \atom as well, and the previous argument
developed for $\atom_5$ applies again here.

In a similar note, observe that $\bigcup_{S\in\Smc(\atom_3)}T_S(\atom_3)=\emptyset$. The reason for this is that $\atom_3$
contains $\bigcirc(a U b)$, which means that every successor node must satisfy $a U b$, which implies that we cannot use
any tuple containing $S_{00}$ or $S_{10}$, as those must violate this formula. But as we have seen already no subset of
$S_2=\{Q_{01},Q_{11}\}$ belongs to $\Smc(\atom_3)$. Finally, $\atom_4$ does not have any successors either. Indeed, every
successor of $\atom_4$ must contain $\neg a\land\neg b$, but in every transition there must be an element which satisfies $b$.

From all this information, we can construct the automaton $\Amc_\psi$. A part of this automaton is depicted in Figure~\ref{fig:autexa}.
\begin{figure}[t!]
\centering
\includegraphics{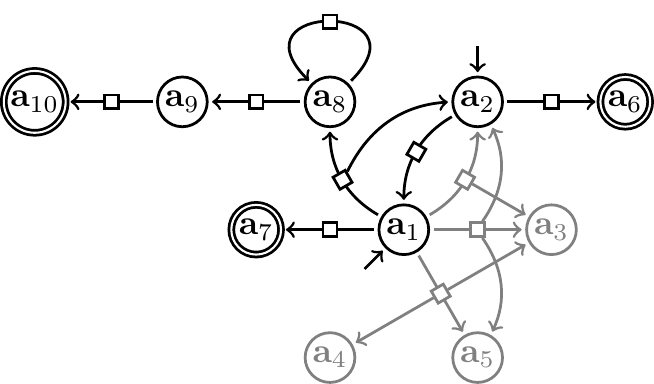}
\caption{A part of the automaton $\Amc_\psi$ with boxes representing hyperedges. $\atom_1$ and $\atom_2$ are initial states.
Double lines represent final states. Gray nodes and edges are those that are removed from the reduced automaton.}
\label{fig:autexa}
\end{figure}
As it can be seen, the atom $\atom_5$ is in fact a bad state, and hence it, along with all transitions leading to it, is removed from
the reduced automaton $\dddot{\Amc_\psi}$. On the other hand, the language accepted by this automaton is not empty, and 
hence $\psi$ is satisfiable. Moreover,
we can follow transitions from an initial state to final states to construct an accepted tree, which will serve as a model of the 
formula $\psi$. 
For example, Figure~\ref{fig:appmodel} shows an abstract template of some models obtained this way.
For any values such that $p_1+p_2=q_1+q_2=1$, and satisfies $p_1,q_1\le 0.6; p_2,q_2\le 0.7$, the figure represents a model. 
\begin{figure}
\centering
\includegraphics{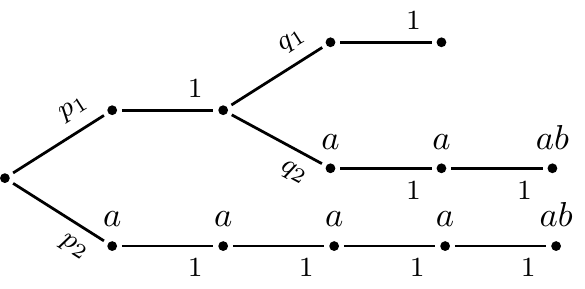}
\caption{Some models of $\psi$.}
\label{fig:appmodel}
\end{figure}
This model is obtained starting from the initial
state $\atom_1$ (which, as a valuation, makes $a$ and $b$ false), and is allowed to make the transition $(\atom_1,\atom_8,\atom_2)$.
The upper branch in the model represents the execution from $\atom_2$, which makes a transition back to $\atom_1$, and then
reuses the same transition $(\atom_1,\atom_8,\atom_2)$, but this time, $\atom_2$ makes a transition to $\atom_6$, which is a final
state. The lower branches show some iterations remaining in $\atom_8$ before making a transition to $\atom_9$ and finally 
reaching $\atom_10$ which is a final state. In models following this pattern, the optimist view on the probability of observing
$\bigcirc\bigcirc\bigcirc a$ is 0.7, obtained by traversing the lower branch. However, if we know that after one timestep $a$ was not
observed, then the probability becomes $0.6\cdot 0.7=0.42$ as obtained from the middle branch. Notice that this discussion is limited
to this class of models only.

We now transform this automaton into the weighted automaton $\Bmc_\psi$ in order to handle the most likely traces. A portion
of this automaton, corresponding to the fragment depicted in Figure~\ref{fig:autexa} of $\Amc_\psi$ is shown in Figure~\ref{fig:waut}.
\begin{figure}
\centering
\includegraphics{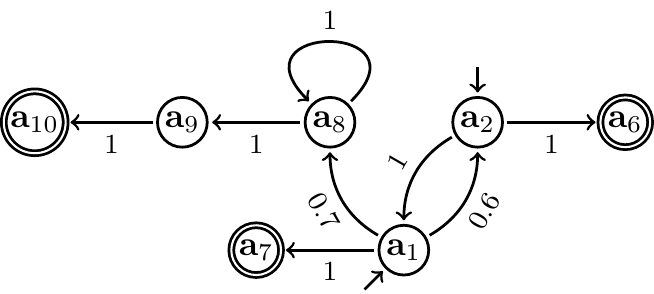}
\caption{$\Bmc_\psi$.}
\label{fig:waut}
\end{figure}
Note that, for example, the transition $(\atom_1,\atom_8,\atom_2)$ from $\Amc_\psi$ gives rise to the transitions $(\atom_1,\atom_8)$
and $(\atom_1,\atom_2)$ with weights $0.7$ and $0.6$ respectively, which are obtained by maximising the potential values of the
systems of inequalities that produce them. However, as mentioned before, these maximisations do not define a probability: there 
is no model that makes both transitions reach this maxima; indeed, in any model of $\Amc_\psi$ the sum of the probabilities of these
transitions needs to always be 1. 

Just by observing this fragment of $\Bmc_\psi$, we can immediately see that $\|\Bmc_\psi\|=1$. Indeed, the runs $\atom_1,\atom_7$
and $\atom_2,\atom_6$ both have weight 1. These runs define traces with a probability 1 of occurring; i.e., $(\emptyset,\{a\})$, and
$(\emptyset,\emptyset)$, respectively. Indeed, notice that both probabilistic constraints in the formula $\psi$ give only an upper
bound on the probability of observing some behaviour. Hence, there are models that assign a probability 0 to both of them (i.e., a
probability 1 of none occurring) as witnessed by these traces. For a prefix $(\emptyset, \{a\}, \{a\}, \{a,b\})$
one must satisfy the probabilistic constraint $(a U b)$, and so the highest probability possible is $0.7$, witnessed by the run
$\atom_1,\atom_8,\atom_9,\atom_{10}$ (with the mlt being exactly that prefix).

To conclude this example, we depict the flattened automaton obtained from  $\Bmc_\psi$ in Figure~\ref{fig:faut}.
\begin{figure}
\centering
\includegraphics{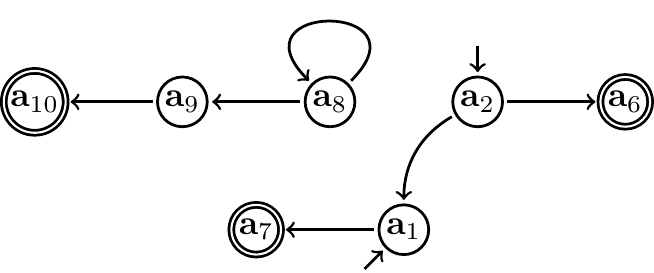}
\caption{$\overline{\Bmc_\psi}$.}
\label{fig:faut}
\end{figure}

\end{document}